\documentclass[11pt]{article}

\usepackage{fullpage}
\usepackage{amsthm,amsmath}  
\usepackage{xspace,enumerate}
\usepackage[utf8]{inputenc}
\usepackage{thmtools}
\usepackage{thm-restate}
\usepackage{authblk}
\usepackage[hypertexnames=false,colorlinks=true,urlcolor=Blue,citecolor=Green,linkcolor=BrickRed]{hyperref}
\usepackage[dvipsnames,usenames]{color}
\usepackage{graphicx}
\usepackage[noadjust]{cite}
\usepackage{microtype}
\usepackage{todonotes}
\usepackage{thmtools}
\usepackage{thm-restate}
\usepackage[capitalise]{cleveref}
\usepackage{algorithmic}
\usepackage[pagewise]{lineno}
\nolinenumbers

\theoremstyle{plain}
\newtheorem{theorem}{Theorem}
\newtheorem{lemma}[theorem]{Lemma} 
\newtheorem{corollary}[theorem]{Corollary}

\theoremstyle{remark}

\usepackage{footnote}
\makesavenoteenv{tabular}
\usepackage{float}
\def \eps{\varepsilon}

\usepackage{url}
\usepackage{comment}
\usepackage{paralist}
\usepackage{enumitem}
\usepackage{caption}

\nolinenumbers

\crefformat{footnote}{#2\footnotemark[#1]#3}

\author[1]{Amir Abboud}
\author[2]{Shon Feller}
\author[2]{Oren Weimann}

\affil[1]{
IBM Almaden Research Center, USA\\
\href{mailto:amir.abboud@ibm.com}{amir.abboud@ibm.com}}

\affil[2]{
University of Haifa, Israel\\
 \href{mailto:shonfeller1@gmail.com}{shonfeller1@gmail.com}, \href{mailto:oren@cs.haifa.ac.il}{oren@cs.haifa.ac.il}
}

\date{}

\title{On the Fine-Grained Complexity of Parity Problems}

\begin{document}

\maketitle

\begin{abstract}

We consider the \emph{parity} variants of basic problems studied in fine-grained complexity.
We show that finding the exact solution is just as hard as finding its parity (i.e. if the solution is even or odd) for a large number of classical problems, including All-Pairs Shortest Paths (APSP), Diameter, Radius, Median, Second Shortest Path, Maximum Consecutive Subsums, Min-Plus Convolution, and $0/1$-Knapsack.

A direct reduction from a problem to its parity version is often difficult to design. Instead, we revisit the existing hardness reductions and tailor them in a problem-specific way to the parity version.
Nearly all reductions from APSP in the literature proceed via the (subcubic-equivalent but simpler) Negative Weight Triangle (NWT) problem.
Our new modified reductions also start from NWT or a non-standard parity variant of it.
We are not able to establish a subcubic-equivalence with the more natural \emph{parity counting} variant of NWT, where we ask if the number of negative triangles is even or odd.
Perhaps surprisingly, we justify this by designing a reduction from the seemingly-harder Zero Weight Triangle problem, showing that parity is (conditionally)  strictly harder than decision for NWT.

\end{abstract}

\section{Introduction}

The blossoming field of fine-grained complexity is concerned with understanding the time complexity of basic computational problems in a precise way. 
The main approach is to hypothesize the hardness of a few core problems and then reduce them to a large number of other problems, establishing tight conditional lower bounds for them.
A cornerstone finding in this field is that there is a class of more than ten problems that are all \emph{subcubic equivalent} to the {\sc All-Pairs Shortest Paths} (APSP) problem, in the sense that if any of them can be solved in truly subcubic $O(n^{3-\eps})$ time (for some $\eps>0$) then all of them can.
Most of the problems in this ``APSP-class'' are related to distance computations in graphs such as computing the radius of the graph or deciding if the graph contains a negative weight triangle (NWT).
In this work, we investigate the fine-grained complexity of the natural \emph{parity versions} of such problems: are they easier, harder, or do they have the same time complexity?
\newpage
\noindent Depending on the problem, the natural parity version could have a different type; let us consider the two main types that will appear in this paper and illustrate them with examples.
\begin{itemize}
\item \textbf{Parity Computation:} The {\sc Radius-Parity} problem asks whether the radius of the graph is even or odd. Similarly, the {\sc APSP-Parity} problem asks to compute, for each pair of nodes, whether the distance between them is even or odd. This type is often natural for optimization problems.
\item \textbf{Parity Counting:} The {\sc NWT-Parity} problem asks if the number of negative triangles in the graph is even or odd, or equivalently, it asks to count the number of negative triangles modulo $2$. Similarly, {\sc SAT-Parity} asks if the number of satisfying assignments to a given formula is even or odd. This type is often natural for decision problems where we are looking for a solution satisfying a certain property\footnote{The parity counting version could be viewed simply as the parity computation version of the counting version of the problem, so the first type could be considered the ``real'' parity version. However, parity counting is widely referred to as the parity version in the literature.
}.
\end{itemize}

The parity computation versions are clearly no harder than the original problem: If we know the radius exactly we also know its parity (if it is even or odd). In fact, sometimes knowing the parity can be much easier than computing the entire answer. For instance, while computing the maximum number of nodes in a matching  requires super-linear time, knowing the parity is trivial (it is always $0$).
On the other hand, parity counting versions can make the problem much harder. A famous example is $2${\sc -SAT}: the decision version takes linear time but the parity version is probably not in P~\cite{valiant2006accidental}. Another example is computing the  permanent of a matrix, which is not expected to be in P (since then $\#$P collapses to P), yet computing the {\em parity} of the permanent is in P (since it is the same as the parity of the determinant). 
Thus, in general, the natural parity version could be easier or harder than the original problem.

Various questions related to parity arose naturally in different contexts in computer science throughout the years.
For instance, the {\sc SAT-Parity} problem played a key role in the proof of Toda's theorem~\cite{TodasTheorem}, which is one of the earliest and most fundamental results in the large body of works on counting complexity~\cite{fortnow1997counting}. There, parity counting problems are extensively studied, being of intermediate complexity between the decision problems and the counting problems (see e.g.~\cite{valiant1979complexity,papadimitriou1982two,valiant1985np,beigel1998np,arvind2002graph,valiant2006accidental,valiant2018some}). Another example is Seidel's $\tilde{O}(n^\omega)$ algorithm~\cite{Seidel} for APSP in unweighted undirected graphs which computes APSP parity as part of the algorithm.
The first type of problems are less studied in terms of worst-case complexity but are morally related to hard-core predicates in cryptography~\cite{vasco2001survey} where it is desirable that the parity (least significant bit) of a function is hard to guess.
The motivation for our work is twofold: first, many of the parity versions are interesting on their own right and we would like to know their complexity, and second, this investigation could lead to a deeper understanding of the structure among the original (non-parity) versions.

\subsection{Our Results}

\paragraph{The APSP Class.} 
Our first set of results concern the APSP equivalence class. We have gone through the problems in this class from the works of~\cite{AbboudEtAl,EricIO,WilliamsNegativeTriangle} and tried to classify the complexity of their parity versions.
Our first theorem shows that, with the notable exception of {\sc Negative Weight Triangle} (NWT), all the parity versions are subcubic-equivalent to APSP and therefore also to the original (non-parity) problems. 
These problems and their parity versions are listed and defined in Table~\ref{table:Table1} together with our results for each of them and where they appear in the paper.

\begin{table}[h!]
\begin{tabular}{|p{2.3cm}|p{5.7cm}|p{7.2cm}|}
\hline 
 Problem & Definition & Complexity \\ \hline
 
 Median & $\min_{u}\sum_{v}d(u,v)$ & \begin{tabular}[c]{@{}l@{}}SE to APSP~\cite{AbboudEtAl},\\ \textbf{Parity is SE to APSP (Sec.~\ref{sec:MedianParity})}\end{tabular} \\ \hline

\begin{tabular}[c]{@{}l@{}}Wiener Index\end{tabular} & $\sum_{u}\sum_{v}d(u,v)$ & \begin{tabular}[c]{@{}l@{}}SE to APSP~\cite{EricIO},\\ \textbf{Parity is SER to APSP} \textbf{(Sec.~\ref{subsec:NTPtoWIP},~\ref{subsec:NTPtoWIPundirected})}\end{tabular} \\ \hline

Radius & $\min_{u}\max_{v}d(u,v)$ & \begin{tabular}[c]{@{}l@{}}SE to APSP~\cite{AbboudEtAl},\\ \textbf{Parity is SE to APSP (Sec.~\ref{sec:RadiusAndDiameterParity})}\end{tabular} \\ \hline

\begin{tabular}[c]{@{}l@{}}Sum of\\ Eccentricities \end{tabular}
& $\sum_{u}\max_{v}d(u,v)$ & \begin{tabular}[c]{@{}l@{}}\textbf{SE to APSP, } \textbf{(Sec.~\ref{sec:sumofeccentricities})},\\ \textbf{Parity is SER to APSP} \textbf{(Sec.~\ref{subsec:NTPtoSoE})}\end{tabular} \\ \hline

\begin{tabular}[c]{@{}l@{}}Integer\\ Betweenness\\Centrality\end{tabular} & \begin{tabular}[c]{@{}l@{}}find the number of vertices pairs \\ with a shortest path passing \\through a given vertex $x$\end{tabular} & \begin{tabular}[c]{@{}l@{}}SE to APSP~\cite{AbboudEtAl},\\ $(1+\eps)$-approx. is SER to Diameter~\cite{AbboudEtAl}\\ \textbf{Parity is SER} \textbf{to APSP}  \textbf{(Sec.~\ref{subsec:NTPtoIntegerBCParity},~\ref{subsec:APSPtoIntegerBCParity})}\end{tabular} \\ \hline

\begin{tabular}[c]{@{}l@{}}Second \\Shortest Path\end{tabular} & \begin{tabular}[c]{@{}l@{}}
given vertices $s,t$, find the length\\  of the second shortest $s$-to-$t$ path
\end{tabular} & \begin{tabular}[c]{@{}l@{}}SE to APSP~\cite{WilliamsNegativeTriangle}, \\ \textbf{Parity is SE to APSP (Sec.~\ref{subsec:ReplacementPath})}\end{tabular} \\ \hline

\begin{tabular}[c]{@{}l@{}}Maximum \\Subarray  \end{tabular}& \begin{tabular}[c]{@{}l@{}}given a matrix, find the maximum \\ total value in a submatrix\end{tabular} & \begin{tabular}[c]{@{}l@{}}SE to APSP~\cite{MaxSubArray,MaxArrayToMatrixMulti},\\ \textbf{Parity is SE to APSP (Sec.~\ref{subsec:MaxSubParity})}\end{tabular} \\ \hline

APSP & Compute all distances $d(u,v)$  & \begin{tabular}[c]{@{}l@{}}\textbf{Parity is SE to APSP (Sec.~\ref{sec:minplusmatrixmul})}\end{tabular} \\ \hline
 
\begin{tabular}[c]{@{}l@{}} Min-Plus \\Matrix\\ Multiplication \end{tabular} & \begin{tabular}[c]{@{}l@{}}
given $n\times n$ matrices $A$ and $B$,\\compute the matrix $C$ where\\ $C[i,j]=\min_k \{A[i,k]+B[k,j]\}$
\end{tabular} & \begin{tabular}[c]{@{}l@{}}SE to APSP (folklore), \\ \textbf{Parity is SE to APSP (Sec.~\ref{sec:minplusmatrixmul})}\end{tabular} \\ \hline

\begin{tabular}[c]{@{}l@{}}Replacement \\Paths \end{tabular} & \begin{tabular}[c]{@{}l@{}}
for every edge $e$ on a shortest\\ $s$-to-$t$  path, find the length of the \\shortest $s$-to-$t$ path that avoids $e$
\end{tabular} & \begin{tabular}[c]{@{}l@{}}SE to APSP~\cite{WilliamsNegativeTriangle}, \\ \textbf{Parity is SE to APSP (Sec.~\ref{subsec:ReplacementPath})}\end{tabular} \\ \hline

\begin{tabular}[c]{@{}l@{}}Negative\\ Weight \\Triangle\end{tabular} & \begin{tabular}[c]{@{}l@{}}determine if there is a triangle \\of total negative weight \end{tabular} & \begin{tabular}[c]{@{}l@{}}SE to APSP~\cite{WilliamsNegativeTriangle}, $(1+\epsilon)$-approx.  Counting \\ is SER to APSP~\cite{ApproximateCountingProblems}.\\ \textbf{Randomized reductions from APSP}\\
\textbf{and 3SUM to Parity and Counting}\\ \textbf{(Sec.~\ref{subsec:ZWTCountingParityToNegTriangleCountingParity},~\ref{subsec:ZWTcountingParity}), Vertex Parity is SER }  \\ \textbf{to APSP}  \textbf{{(Sec.~\ref{subsec:NTtoNTP})}}\end{tabular}
\\ \hline

\begin{tabular}[c]{@{}l@{}}Zero\\ Weight\\ Triangle  \end{tabular}& \begin{tabular}[c]{@{}l@{}}determine if there is a triangle of \\total zero weight\end{tabular} & \begin{tabular}[c]{@{}l@{}}Reduction from APSP and 3SUM~\cite{ZWT},\\ \textbf{Randomized reduction to Parity and}\\ \textbf{Vertex Parity} \textbf{(Sec.~\ref{subsec:NTtoNTP},~\ref{subsec:ZWTcountingParity})}
\end{tabular} \\ \hline

\end{tabular}
\vspace{0.1in}
\caption{The APSP class: problem definitions, known results, and our results (in bold). We denote subcubic equivalent as SE and as SER when it is under a randomized reduction.
The first seven problems output a single value and their parity version computes the parity of this value. The next three problems output multiple values and their parity version computes the parity of every value. The last two problems are the only parity counting problems, in which we distinguish between the parity version (asking for the parity of the number of such triangles) and the vertex parity version (asking for the parity of the number of vertices that belong to such triangles).
 \label{table:Table1}}
\end{table}

\clearpage

\begin{theorem}\label{thm:main}
The following problems are subcubic-equivalent:
\begin{itemize}
\item All-Pairs Shortest Paths and its parity computation,
\item Min-Plus Matrix Multiplication and its parity computation,
\item Radius and its parity computation,
\item Median and its parity computation,
\item Wiener Index and its parity computation,
\item Replacement Paths and its parity computation,
\item Second Shortest Path and its parity computation,
\item Vertex in Negative Weight Triangle and its parity computation,
\item Integer Betweenness Centrality and its parity computation,
\item Maximum Subarray and its parity computation,
\item Sum of Eccentricities\footnote{This natural problem was not considered before to our knowledge, but it is closely related to Median, Radius, and the other distance computation problems.} and its parity computation.
\end{itemize}
\end{theorem} 

This adds more than ten natural problems to the APSP-equivalence class. 
For all problems in Theorem~\ref{thm:main}, the reduction from the parity version to the original problem is straightforward (since they are parity computation rather than parity counting problems), while the reduction in the other direction is not. 
For instance, it is not at all clear how to establish the hardness of {\sc Median-Parity} by reducing from {\sc Median} to it. 
Instead, we find it much more convenient to start from NWT which is the canonical APSP-complete problem and the starting point for nearly all APSP-hardness reductions.

Some of our results take the known reductions from NWT and modify them to establish the hardness of the parity versions, e.g. for {\sc Median-Parity}. Notably, reductions of this kind are deterministic.
For other parity problems such as {\sc Wiener-Index} it is more convenient to reduce from a parity version of NWT.
However, (the most natural) {\sc NWT-Parity} is a parity counting problem which makes it seem harder than NWT and therefore inappropriate as a starting point for reductions.
Instead, we identify a different variant that we call {\sc NWT-Vertex-Parity} (asking if the number of vertices that belong to a negative triangle is even or odd) which turns out to be subcubic-equivalent to NWT and a very useful intermediate problem. 
Reductions of this kind seem to require randomization.

Finally, we investigate the intriguing {\sc NWT-Parity} problem.
This is the modulo 2 version of the {\sc NWT-counting}  problem (asking for the number of negative triangles) that was recently studied by Dell and Lapinskas~\cite{ApproximateCountingProblems} in their work on the fine-grained complexity of approximate counting.
With standard subsampling techniques, one can show that NWT reduces to {\sc NWT-Parity}. But are they subcubic-equivalent?
We show that such an equivalence would imply breakthroughs in fine-grained complexity, therefore suggesting that the parity version is strictly harder. 
Our next theorem shows that {\sc NWT-Parity} can solve a problem that is considered strictly harder than APSP: the problem of deciding whether a graph has a {\sc Zero Weight Triangle} (ZWT).
As discussed below, if the same reduction can be shown between the original (non-parity) problems it would be a major breakthrough.

\begin{theorem}\label{thm2}
There is a deterministic subcubic-reduction from the Zero Weight Triangle Parity problem to the Negative Weight Triangle Parity problem.
\end{theorem}

The ZWT problem is considered one of the ``hardest'' $n^3$-problems since a subcubic algorithm for it would refute \emph{two} of the main conjectures in fine-grained complexity: it would give a subcubic algorithm for APSP and a subquadratic algorithm for 3SUM~\cite{ZWT,patrascu2010towards,WilliamsNegativeTriangle}.
The {\sc 3SUM} Conjecture states that we cannot decide in truly subquadratic $O(n^{2-\eps})$ time if among a set of $n$ numbers there are three that sum to zero. The class of problems that are 3SUM-hard contains dozens of problems mostly from computational geometry (see~\cite{ComputationalGeo3SUM,King2004} for a partial list), but also in other domains, e.g. \cite{abboud2014popular,kopelowitz2016higher,patrascu2010towards}.
One of the central open questions in the field is whether the APSP class and the 3SUM class can be unified; in particular, whether APSP is 3SUM-hard.
One way to prove this is to reduce ZWT to APSP, and our result shows that NWT-Parity to NWT suffices:

\begin{corollary}
If the Negative Weight Triangle Parity problem is subcubic-equivalent to the Negative Weight Triangle problem, then APSP is 3SUM-hard. 
\end{corollary} 

A more quantitative reason for supposing that ZWT is harder than NWT is that their current upper bounds, while all mildly-subcubic, are significantly far apart.
All the problems in the APSP equivalence class can be solved in $n^3/2^{\Omega(\sqrt{\log n})}$ time~\cite{WilliamsAPSP}, which is faster than $O(n^3/\log^c{n})$ for all $c>0$. 
For ZWT, on the other hand, nothing better than $O(n^3/\log^c{n})$ is known for a small $c\leq 2$, and even small improvements would lead to a faster mildly subquadratic algorithm for 3SUM beating the current fastest $O(n^2 (\log\log n)^2/\log n)$~\cite{Best3Sum}.
Dell and Lapinskas~\cite{ApproximateCountingProblems} achieve an $n^3/2^{\Omega(\sqrt{\log n})}$ upper bound for the approximate counting version of NWT but not for exact counting. We show that even for the parity version such a result has breakthrough consequences for 3SUM.
For NWT, this (conditionally) separates the counting and parity versions from the decision and approximate counting versions.

\begin{table}[h!]
\begin{tabular}{|p{2.3cm}|p{5.7cm}|p{7.2cm}|}
\hline 
 Problem & Definition & Complexity \\ \hline
 Diameter & $\max_{u}\max_{v}d(u,v)$ & \begin{tabular}[c]{@{}l@{}}
\textbf{Parity is SE to Diameter (Sec.~\ref{sec:RadiusAndDiameterParity})}\end{tabular} \\ \hline

Maximum Row Sum & $\max_{u}\sum_{v}d(u,v)$ & \begin{tabular}[c]{@{}l@{}}\textbf{Reduction from Co-Negative Triangle} \\ \textbf{to Parity (Sec.~\ref{sec:maxrowsum})}\end{tabular} \\ \hline

\begin{tabular}[c]{@{}l@{}}Reach \\Centrality\end{tabular} & \begin{tabular}[c]{@{}l@{}}compute the maximum distance \\between a given vertex $x$ and the\\closest endpoint of any shortest\\path passing through $x$\end{tabular} & \begin{tabular}[c]{@{}l@{}}SE to Diameter~\cite{AbboudEtAl},\\ \textbf{Parity is SE to Diameter} \textbf{(Sec.~\ref{sec: ReachCentralityParity})}\end{tabular}\\ \hline

\begin{tabular}[c]{@{}l@{}} $0/1$-Knapsack \end{tabular} & \begin{tabular}[c]{@{}l@{}}
given items $(w_{i},v_{i})$ and a weight\\ $t$, find a subset  $I$ that maximizes \\$\sum_{i\in I} v_{i}$ subject to $\sum_{i\in I} w_{i} \leq t$
\end{tabular} & \begin{tabular}[c]{@{}l@{}}SQER to Min-Plus Convolution, variants \\are SQE to Min-Plus Convolution~\cite{MinPlusConvEquivalences,DynamicProgramingEquivalences}, \\\textbf{Parity is} \textbf{SQER to Min-Plus }\\\textbf{Convolution, variants are SQE to} \\\textbf{Min-Plus Convolution (Sec.~\ref{sec:minplusknapsackparity})}\end{tabular} \\ \hline

\begin{tabular}[c]{@{}l@{}} Tree Sparsity \end{tabular} & \begin{tabular}[c]{@{}l@{}}
given a node-weighted tree, find \\ the maximum weight of a subtree \\ of size $k$
\end{tabular} & \begin{tabular}[c]{@{}l@{}} SQE to Min-Plus Convolution~\cite{MinPlusConvEquivalences,TreeSparsityToMinPlusConv},\\ \textbf{Parity is} \textbf{SQE to Min-Plus}\\ \textbf{Convolution (Sec.~\ref{sec:TreeSparsity})}\end{tabular} \\ \hline

\begin{tabular}[c]{@{}l@{}} \hspace{-0.5mm}Min-Plus \\Convolution \end{tabular} & \begin{tabular}[c]{@{}l@{}}
given $n$-length vectors $A$ and $B$,\\compute the vector $C$ where\\ $C[k]= \min_{i+j=k} \{A[i]+B[j]\}$
\end{tabular} & \begin{tabular}[c]{@{}l@{}}Reduction to APSP and 3SUM~\cite{MinPlusConvEquivalences,ConvolutionsAndNecklaces},\\ \textbf{SQE to Parity} \textbf{(Sec.~\ref{subsec:MinplusConvToParity})}\end{tabular} \\ \hline

\begin{tabular}[c]{@{}l@{}} Maximum \\Consecutive \\Subsums \end{tabular} & \begin{tabular}[c]{@{}l@{}}
given an $n$-length vector $A$,\\compute the vector $B$ where\\ $B[k]= \max_{i} \{\sum_{j=0}^{k-1}A[i+j]\}$
\end{tabular} & \begin{tabular}[c]{@{}l@{}}SQE to Min-Plus Convolution~\cite{MinPlusConvEquivalences,MaxConsSubsum}, \\ \textbf{Parity is SQE to Min-Plus}\\ \textbf{Convolution (Sec.~\ref{subsec:MaxConsSumToParity})}\end{tabular} \\ \hline

\begin{tabular}[c]{@{}l@{}}Co-Negative \\Triangle\end{tabular} & \begin{tabular}[c]{@{}l@{}}find a vertex that does not \\belong to any negative triangle\end{tabular} & \begin{tabular}[c]{@{}l@{}}Reduction to Diameter~\cite{Complementary}, \textbf{Reduction to }\\\textbf{Maximum Row Sum Parity} \textbf{(Sec.~\ref{sec:maxrowsum})}\end{tabular} \\ \hline

\end{tabular}
\vspace{0.1in}
\caption{The other (non-APSP) problems. We denote subcubic (subquadratic) equivalent as SE (SQE) and as SER (SQER) when it is under a randomized reduction.
As before, the parity version of problems that output a single (multiple) value(s) computes the parity of this value (all these values). The last problems is the only parity counting problem, in which the parity version asks for the parity of the number of vertices that do not belong to any negative triangle.
\label{table:Table2}}
\end{table}

\paragraph{Other Classes.}
In our second set of results we ask whether parity computation problems are as hard also for problems that are outside the APSP class.
We have gone through other fine-grained complexity results from the works of~\cite{MinPlusConvEquivalences,TreeSparsityToMinPlusConv,DynamicProgramingEquivalences,MaxConsSubsum} and tried to establish the same results for the parity versions.
All problems we consider are defined in Table~\ref{table:Table2} together with our results and where they appear in the paper.
The general message is that, in all cases we considered, the same hardness reductions (if modified carefully) can establish the hardness of the (seemingly easier) parity version as well.
We mention a few concrete examples.

In the context of distance computations in graphs, the central open question is whether the {\sc Diameter} problem is subcubic equivalent to APSP. 
Meanwhile, {\sc Diameter} has its own (smaller) equivalence class which includes problems such as {\sc Reach Centrality}~\cite{AbboudEtAl}.
We prove that {\sc Diameter-Parity} and {\sc Reach Centrality-Parity} are subcubic equivalent to {\sc Diameter}.

Another interesting problem in fine-grained complexity whose importance is rapidly increasing is the {\sc Min-Plus Convolution} problem \cite{MinPlusConvEquivalences}. 
The na\"ive algorithm for this problem runs in $O(n^2)$ time, and a truly subquadratic $O(n^{2-\eps})$ algorithm is conjectured to be impossible.
This problem is one of the easiest $n^2$ problems since it can be reduced to both APSP (i.e. a subcubic algorithm for APSP yields a subquadratic algorithm for {\sc Min-Plus Convolution}) and to 3SUM (an opposing situation to that of ZWT). This means that all of the APSP and 3SUM lower bounds can be based on this conjecture, but also that {\sc Min-Plus Convolution} is unlikely to be equivalent to either of them (as it would imply a unification of the classes).
Recently, a few other problems have been shown to be harder, e.g. \cite{ACAK2020}, or subquadratic-equivalent to it, e.g. {\sc Maximum consecutive subsums}~\cite{MaxConsSubsum,MinPlusConvEquivalences},  {\sc $0/1$-Knapsack}~\cite{DynamicProgramingEquivalences,MinPlusConvEquivalences}, and a $(1+\eps)$-approximation for {\sc Subset Sum}~\cite{SubsetSums}.
We prove that these equivalences hold for the parity versions as well (except the latter problem for which we did not find a natural parity version).
Our reduction from the {\sc Maximum consecutive subsums} problem to its parity version in Section~\ref{subsec:MaxConsSumToParity} is quite involved and it uses specific properties of the addition operator. One can obtain such a reduction indirectly and more easily via {\sc Min-Plus Convolution}, however, we believe that our reduction gives more insight into the problem and into the usage of the addition operator.

\clearpage
\begin{theorem}\label{thm:conv}
The following problems are subquadratic-equivalent:
\begin{itemize}
\item Min-Plus Convolution and its parity computation,
\item Maximum Consecutive Subsums and its parity computation,
\item $0/1$-Knapsack and its parity computation,
\item Tree Sparsity and its parity computation.
\end{itemize}
\end{theorem}

\subsection{Related Work}

While parity counting problems are extensively studied in classical  complexity theory, the parity computation problems seem to have received less attention.
In many cases, the standard NP-hardness reduction from SAT gives instances in which the solution is always either $k$ or $k-1$, which directly implies the NP-hardness of the parity version as well.
Some of the results in fine-grained complexity also have this property. 
For example, the quadratic hardness result for {\sc Diameter} in sparse graphs~\cite{RV13} shows that it is hard to distinguish diameter $2$ from $3$ and immediately gives the same lower bound for parity.
However, for many other problems, such as the ones we consider, this is not the case and a careful problem-specific treatment is required.

Theorem~\ref{thm2} and its corollaries conditionally separate NWT from its parity and counting versions. 
Such separations are famously known in classical complexity, e.g. for $2${\sc -SAT}~\cite{valiant2006accidental}. 
In fine-grained complexity, a (conditional) separation for a variant of the {\sc Orthogonal Vectors} problem between near-linear time decision~\cite{OVDiscreteStructures} and quadratic time exact counting~\cite{OrthogonalVectorCounting} was recently achieved. 
Notably, the approximate counting version is also in near-linear time~\cite{ApproximateCountingProblems}  and the parity version is open.

The parity counting version of the Strong Exponential Time Hypothesis was studied in a seminal paper on the fine-grained complexity of NP-hard problems~\cite{OnProblemsAsHardAsCNFSAT}.
The central question left open in that paper (and is still wide open, see~\cite{AbboudSetCover}) is whether SAT can be reduced to Set-Cover in a fine-grained way; interestingly, the authors have shown such a reduction for the parity counting versions.

Exact and approximate counting problems have received a lot of attention in parameterized \cite{flum2004parameterized,curticapean2019counting} and fine-grained complexity \cite{dell}.
In a recent development, the {\sc  $k$-Clique} counting problem was shown to have worst-case to average-case reductions~\cite{BBB19,Goldreich-RonFOCS2018}.
It is likely that our result for {\sc NWT-Parity} can be extended to {\sc Negative Weight $k$ Clique Parity} showing that it is as hard as {\sc Zero Weight $k$ Clique}. 
The decision version of the latter problem was used as the basis for public-key cryptography schemes~\cite{FineGrainedCryptography}.

Due to the large amount of works on APSP-hardness and equivalences we did not manage to exhaustively enumerate all of them and investigate the complexity of the parity versions, e.g. for problems on stochastic context-free grammars~\cite{ContextFreeGrammars} or dynamic graphs~\cite{abboud2014popular,roditty-zwickESA04}. Still, we expect that the ideas in this work can be extended to show the hardness of those parity computation problems as well. 

Besides parity computation and parity counting, there is a third natural type of parity problems where we take a problem and replace one of the operations (e.g. summation) with a parity. For example, the 3XOR problem is a variant of 3SUM where we are given a set of $n$ binary vectors of size $O(\log n)$ and are asked if there is a triple whose bit-wise XOR is all zero. 3XOR is the subject of study of several papers \cite{bouillaguet2019faster,bouillaguet2018revisiting,dietzfelbinger2018subquadratic} and it seems just as hard as 3SUM but a reduction in either direction has been elusive \cite{jafargholi2016mathrm}.

\subsection{Preliminaries}

In all graph problems we assume that the graphs have $n$ nodes and $O(n^2)$ edges. In all the weighted problems we consider, we assume the weights are integers in $[-M,M]$ (and generally it is assumed that $M=poly(n)$). 

Intuitively, a fine-grained reduction \cite{WilliamsNegativeTriangle,FineGrainedQuestions} from problem A with current upper bound $O(n^a)$ to problem B with current upper bound $O(n^b)$ is a Turing-reduction proving that if B is solvable in time $O(n^{b-\eps})$, for some $\eps>0$, then A is solvable in time $O(n^{a-\eps'})$, for some $\eps'>0$.
More formally, an $(a,b)$-fine-grained reduction from A to B is a (possibly randomized) algorithm solving A on instances of size $n$ using $t$ calls to an oracle for B on instances of sizes $n_1,\ldots,n_t$, such that for all $\eps>0$: $\sum_{i=1}^t (n_i)^{b-\eps} \leq  n^{a-\eps'}$ for some $\eps'>0$.
In this paper, unless otherwise stated, we assume that the reduction is randomized.
A $(3,3)$-fine-grained reduction is called a subcubic-reduction and two problems are called subcubic-equivalent if there are subcubic-reductions in both ways. Similarly, two problems are subquadratic-equivalent if there are $(2,2)$-fine-grained reductions between them in both ways.

%Our results on the APSP class are summarized in Table~\ref{table:Table1} and our results on the other problems are summarized in Table~\ref{table:Table2}. 

\section{APSP to Median Parity} \label{sec:MedianParity}
In this section, we show a subcubic reduction from the {\sc Negative Weight Triangle} problem (hence also from APSP~\cite{WilliamsNegativeTriangle}) on a directed graph $G$ with integral edge weights in $[-M,M]$ to {\sc Median Parity}. 
We first describe the reduction of~\cite{AbboudEtAl} from {\sc Negative Weight Triangle} to {\sc Median} and then modify it to become a reduction to {\sc Median Parity}. 

\subsection{Negative Weight Triangle to Median~\cite{AbboudEtAl}}\label{sec:Negative Weight Triangle to Median}
The instance $G'$ to the {\sc Median} problem (illustrated in Figure~\ref{fig:MedianParityGraph}) is an undirected graph constructed as follows.
First, for any two (not necessarily different) vertices $u,v$ if there is no edge $(u,v)$ in $G$ then we add an edge $(u,v)$ of weight $w(u,v)=4M$ to $G$ (this will not form a new negative triangle). Each vertex $u$ of $G$ has five copies in $G'$ denoted $u_A,u_B,u_{B'},u_{C},u_{C'}$. Let $H$ be a sufficiently large number (say $H= 100M$). For any two (not necessarily different) vertices $u,v$ of $G$ we add the following edges to $G'$: $(u_A,v_B)$ of weight $3H+w(u,v)$, $(u_A,v_{B'})$ of weight $3H-w(u,v)$, $(u_A,v_C)$ of weight $6H-w(v,u)$\footnote{\label{VerticesOrder2}Notice the different order of the vertices.},  $(u_A,v_{C'})$ of weight $3H+w(v,u)$\cref{VerticesOrder2}, $(u_A,v_A)$ of weight $H$, and $(u_B,v_C)$ of weight $3H+w(u,v)$.

\begin{figure}[H]
    \centering
    \includegraphics[scale=0.3]{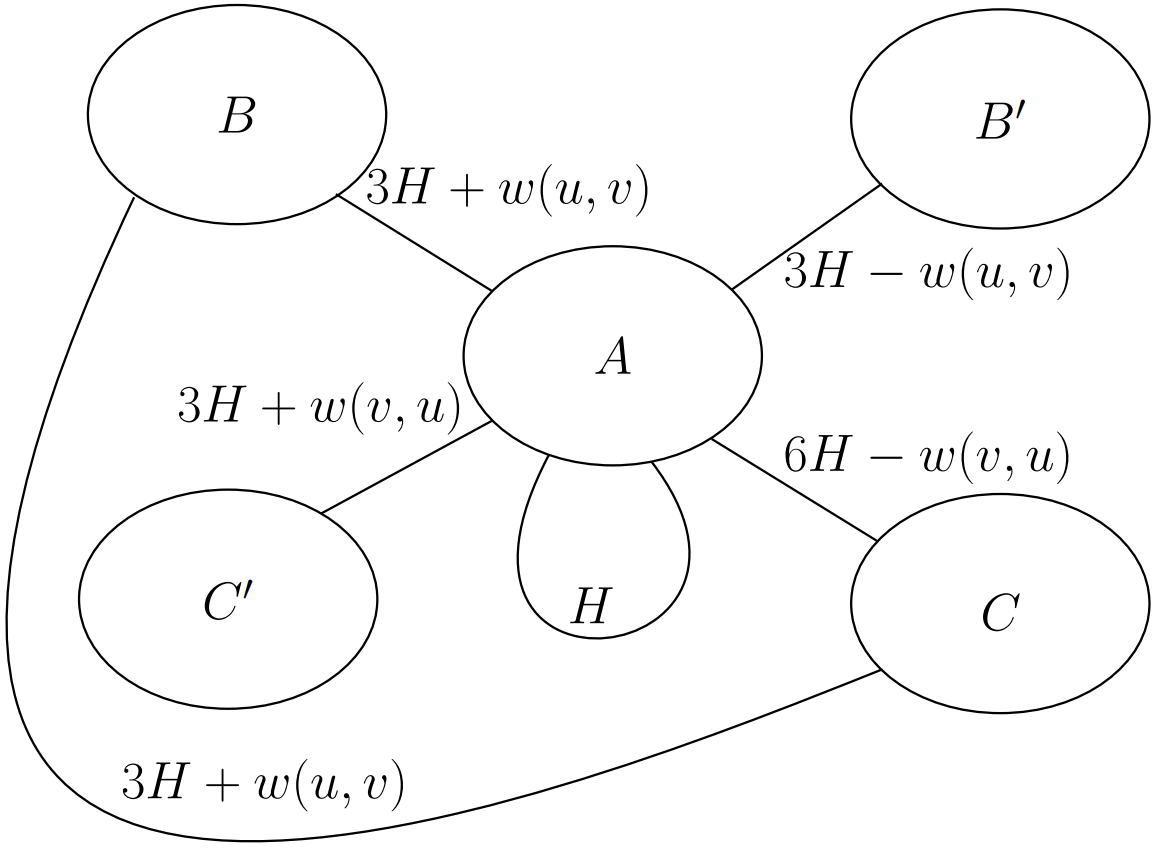}
    \caption{The graph $G'$ in the reduction from Negative Weight Triangle to Median.}
    \label{fig:MedianParityGraph}
\end{figure}

\begin{lemma}[\cite{AbboudEtAl}]
$G$ does not contain a negative triangle iff the median of $G'$ is $(16n-1)H$.
\end{lemma}
 \begin{proof}
Consider first a vertex $u_{X}$ with $X\neq  A$. We claim that the sum of distances $\sum_{v\in V(G')}{d_{G'}(u_X,v)}$ is at least $(19n-5)H$. To see this, first observe that the sum is minimized when $X=B$. This is because shortest paths from vertices in $B'$ and $C'$ go through $A$, and because every $C$-to-$A$ distance is larger than any $B$-to-$A$ distance by at least $H$. We therefore focus on $X=B$:  
The distance from $u_{B}$ to $u_{B}$ is zero and the distance from $u_{B}$ to $v_{B}$ (for $v\neq u$) is at least $5H$ (since $H$ is large enough, $3H+w(u,t)+3H+w(t,v)>5H$), so the sum of distances from $u_{B}$ to all vertices of $B$ is $5H(n-1)$. Similarly, for every vertex $v$ of $G$, the distances from $u_{B}$ to $v_{A},v_{B'},v_{C},v_{C'}$ are at least $2H,5H,2H,5H$ respectively. Overall, the sum of distances from $u_{B}$ is at least $(5n-5)H + 2nH + 5nH + 2nH + 5nH =(19n-5)H$.

Next consider a vertex $u_{A}$. Let $F(u,v)= \min\{0, \min_{t\in V(G)}\{w(v,u)+w(u,t)+w(t,v)\}\}$.
Observe that $F(u,v)= 0$ if the edge $(v,u)$ is not part of any negative triangle in $G$, and $F(u,v)< 0$ otherwise. We claim that the sum of distances $\sum_{v\in V(G')}{d_{G'}(u_{A},v)}$ is exactly $ (16n-1)H+ \sum_{v\in V(G)} F(u,v)$.
To see this, consider the distances from $u_{A}$. Distances to $v_{A},v_{B}$, and $v_{B'}$ are $H$ (for $v\neq u$), $3H+w(u,v)$, and $3H-w(u,v)$ respectively. Over all such vertices the sum of the distances is therefore $(n-1)H+ 6nH = (7n-1)H$. The distance to $v_{C'}$ is $3H+w(v,u)$ and the distance to $v_{C}$ is the minimum between $6H-w(v,u)$ (using a single edge) and $3H+w(u,t)+3H+w(t,v)$ for some vertex $t$ (using two edges, through some $t_{B}$). Summing those two distances together, we get $9H+w(v,u)+\min_{t\in V(G)} \{-w(v,u),w(u,t)+w(t,v)\}=9H+F(u,v)$. Overall, we get that $\sum_{v\in V(G')}{d_{G'}(u_{A},v)} = (16n-1)H + \sum_{v\in V(G)} F(u,v)$ as claimed. This implies that the median vertex must come from $A$ and that the median value is $(16n-1)H$ iff every $F(u,v)=0$ (i.e. $G$ does not contain a negative triangle).
 \end{proof}

\subsection{Negative Weight Triangle to Median Parity}\label{subsec:NegativeTriangleToMedianParity}
We now modify the above reduction so that it reduces to {\sc Median Parity} instead of {\sc Median}. We assume $n$ is odd (otherwise add an isolated vertex to $G$). Let $Med$ be the value of the median of $G'$. We multiply all the edge weights of $G'$ by $4n$ (notice that this multiplies the median value $Med$ by $4n$). We do this in order to make sure that small changes in edge weights would not change any shortest path, and also to make sure that subtracting $n$ from distance sums would not change the median vertex.

We show how to find the median of $G'$ using $O(\log n)$ executions of {\sc Median Parity}: Given a set of vertices $T\subseteq A$ (initialized to be $A$), pick an arbitrary subset $S$ of $T$ of half of its size. Temporarily (i.e restore weights at the end of the iteration) subtract $1$ from all the $S$-to-$B$ and $S$-to-$C$ edges and add $1$ to all the $S$-to-$B'$ edges. Now solve {\sc Median Parity} on $G'$. If the median value is odd, set $T\gets S$.  If the median value is even, set $T\gets T/S$. We continue recursively for $O(\log n)$ steps until $T$ contains a single vertex. We then check if this vertex participates in a negative triangle in $G$.

For the correctness of the above procedure, inductively assume that $T$ contains the median vertex of $G'$. Notice that the temporary changes to the edge weights do not change the identity of shortest paths in $G'$, only their value. In particular, the sum of distances from every vertex $u_{A}\in S$ decreases exactly by $n$, and for any vertex  $u_{A}\in A\backslash S$ the sum remains the same. To see this, consider first a vertex $u_{A}\in S$. The sum of its distances to any $v_{B}$ and $v_{B'}$ remains the same (one is larger by 1 and one is smaller by 1) and its distance to $v_{C}$ is decreased by $1$ (recall that the shortest path is either the direct edge $(u_{A},v_{C})$ or two edges $(u_{A},t_{B}),(t_{B},v_{C})$). Therefore, the sum of distances from $u_{A}\in S$ to all vertices of $G'$ decreases by exactly $n$. As for vertices in $u_A \in A \backslash S$, we do not change weights of edges in their shortest paths so their sum of distances is unchanged.

If the median is from $S$, then its sum of distances in $G'$ was originally $Med$. Since we multiplied  the edge weights by $4n$ and subtracted $n$ from its sum, the median value is now $4n\cdot Med-n$. This value is odd and indeed we set $T\gets S$.
If on the other hand the median is not from $S$, then the sum of distances from any vertex of $S$ was originally at least $Med+1$, and is therefore now at least $4n\cdot(Med+1)-n$. This value is strictly bigger than the value $4n\cdot Med$ of  the median. The value $4n\cdot Med$ is even and indeed we set $T\gets T/S$.

\section{Negative Triangle Vertex Parity} \label{sec:NTP}
In this section, we show that {\sc Negative Triangle Vertex Parity} (finding if the number of vertices that belong to a negative triangle is odd or even) is subcubic equivalent to {\sc APSP} under randomized reductions. We then use {\sc Negative Triangle Vertex Parity} in order to establish a subcubic equivalence with the Parity versions of {\sc Wiener Index}, {\sc Sum of Eccentricities}, and {\sc Integer Betweenness Centrality}.

\subsection{APSP to Negative Triangle Vertex Parity}\label{subsec:NTtoNTP} 
We show a probabilistic (one side error) reduction from {\sc Negative Weight Triangle} to {\sc Negative Triangle Vertex Parity} (NTVP). Given a directed instance of NWT $G$, we can turn it into a tripartite graph by adding vertices $v_1, v_2, v_3$ for every $v$ in $V(G)$, as well as adding the edges $(u_1, v_2),  (u_2, v_3),  (u_3, v_1)$ for every $(u,v)$ in $E(G)$. It holds that there is a negative triangle in $G$ iff there is a negative triangle in the tripartite graph. Furthermore one could replace the directed edges in the tripartite graph with undirected ones. Given an undirected NWT instance $G$, we create an undirected NTVP instance $G'$ as follows: Choose $V_{1}\subseteq V(G)$ uniformly, and let $\overline{V}=V(G)\backslash V_{1}$. For every $u_{1}\in V_{1}$ we add a vertex $u_{2}$, and let the union of all $u_{2}$ vertices be $V_{2}$. For every edge $(u_{1},v_{1})$ in $V_{1}\times V_{1}$ we add the edge $(u_{2},v_{2})$ and for every edge $(u_{1},v')$ in $V_{1}\times\overline{V}$ we add the edge $(u_{2},v')$.
% and for every edge $(v',u_{1})$ in $\overline{V}\times V_{1}$ we add the edge $(v',u_{2})$.
Notice that the graph induced by $\overline{V}\cup V_{2}$ is $G$, and the same for $\overline{V}\cup V_{1}$. 

Since there are no edges between $V_{1}$ and $V_{2}$, every triangle is either in $\overline{V}\cup V_{1}$ or in $\overline{V}\cup V_{2}$. Furthermore, for every vertex $u_1\in V_1$, if $u_1$ belongs to a negative triangle in $G$ then both $u_1$ and $u_2$ belong to negative triangles in $G'$, thus contributing $2$ (even) to the parity NTVP($G'$) of the number of vertices that belong to a negative triangle in $G'$. Therefore, vertices in $V_{1}$ do not affect the parity NTVP($G'$). In other words, NTVP($G'$) is the parity of vertices in $\overline{V}$ with a negative triangle. 
If $G$ contains a negative triangle, then the probability of odd NTVP($G'$) is exactly $1/2$ (since each vertex with a negative triangle is chosen to be in $\overline{V}$ with probability $1/2$). If $G$ does not contain a negative triangle, then the probability of even NTVP($G'$) is exactly $1$. By repeating this process $O(\log n)$ times we can amplify the probability of success to $1- 1/n^{c}$ for any constant $c$. 

We remark that the above reduction can also be used to reduce {\sc Zero Weight Triangle} to its vertex parity version.

\subsection{Negative Triangle Vertex Parity to Wiener Index Parity (Directed)}\label{subsec:NTPtoWIP}
We handle the directed case here and the undirected case in Section~\ref{subsec:NTPtoWIPundirected}. Assume $n$ is even by adding a vertex with no negative triangles, if needed. The reduction graph $G'$ is constructed as in~\cite{AbboudEtAl,WilliamsNegativeTriangle} (see Figure~\ref{fig:ntptowip}): Each vertex $u$ of $G$ has five copies in $G'$ denoted $u_S,u_{A},u_{B},u_{C},u_{D}$. Let $H$ be a sufficiently large {\em even} number (say $H= 100M$). For every $(X,Y)\in \{(A,B),(B,C),(C,D)\}$ and $u,v\in V(G)$, add the edge $(u_{X}, v_{Y})$ with weight $2H+2w(u,v)$. For every $u\neq v\in V(G)$, add an edge $(u_{A}, v_{D})$ with weight $5H$. For every $u\in V(G)$, we add the edge $(u_{S},u_{A})$ with weight $H+1$ and the edge $(u_{S},u_{D})$ with weight $7H$. Turn $G'$ into a clique by replacing any missing edge with an edge of weight $16H$.
\begin{figure}[H]
    \centering
    \includegraphics[scale=0.35]{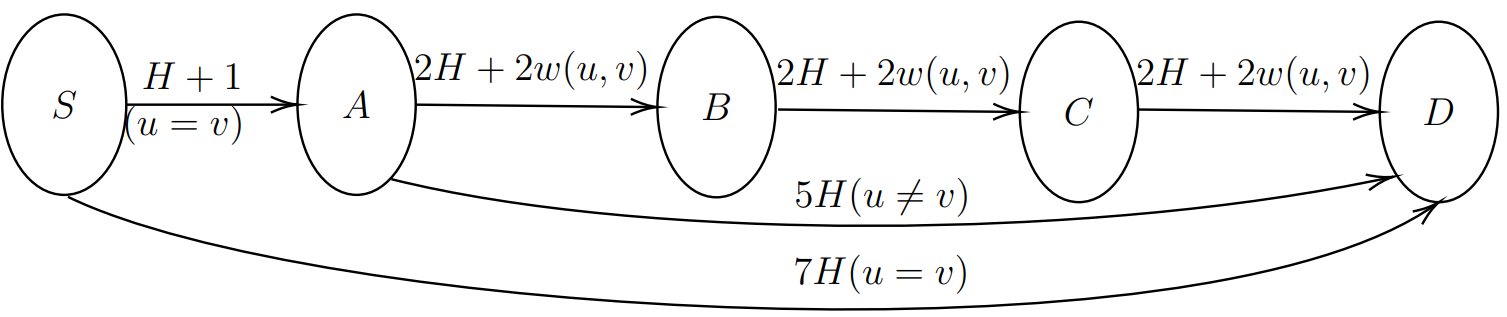}
    \caption{The graph $G'$ in the reduction from Negative Triangle Vertex Parity to Wiener Index Parity. Edges of weight $16H$ are absent.}
    \label{fig:ntptowip}
\end{figure}

We now show that the Negative Triangle Vertex Parity of $G$ (NTVP($G$)) is equal to the Wiener Index Parity of $G'$ (WIP($G'$)). Since $H$ is an even number, the only edges in $G'$ that have an odd length are the $(u_S,u_A)$ edges of length $H+1$. Therefore, the parity WIP($G'$) is determined by the $S$ to $A\cup B \cup C \cup D$ distances.

First observe that the sum of distances from $S$ to $A\cup B \cup C$ is even. This is because for any vertex $u_S$ in $S$ the following shortest paths consist of a single edge of weight $H+1$: the  $u_S$-to-$v_A$ (for $v=u$) path, the $u_S$-to-$v_B$ (for $v=u$ or $v\neq u$) paths, and the $u_S$-to-$v_C$ (for $v=u$ or $v\neq u$) paths. Thus, the total number of odd edges in the sum of distances from $S$ to $A\cup B \cup C$ is $n(2n+1)$, which is even since $n$ is even.

It remains to consider the distances from $S$ to $D$. 
For $u\neq v$, $d_{G'}(u_{S},v_{D})=6H+1$, and the sum of such distances is $n(n-1)(6H+1)$ (even).
We are left with the sum of distances $ d_{G'}(u_{S},u_{D})$. If $u$ belongs to a negative triangle in $G$ and $k$ is the minimal weight of such cycle then $d_{G'}(u_{S},u_{D})= 7H+2k+1$ (odd). If $u$ does not belong to any negative triangle then $d_{G'}(u_{S},u_{D})= 7H$ (even). Therefore the sum of distances is odd iff there is an odd number of vertices belonging to a negative triangle.

\subsection{Negative\! Triangle Vertex Parity to Wiener Index\! Parity\! (Undirected)}\label{subsec:NTPtoWIPundirected}
In undirected graphs, to avoid a trivial  Wiener Index Parity of 0, the Wiener Index is defined as the sum of $d(u,v)$ over every unordered (rather than ordered) pair $\{u,v\}$. 

We assume that every triangle has odd length by multiplying every edge-weight by $4$ and adding $1$ (this preserves the sign of negative and non-negative triangles). We construct a graph $G'$ similarly to~\cite{AbboudEtAl,WilliamsNegativeTriangle} and to Section~\ref{subsec:NTPtoWIP} but the approach differs in the analysis of correctness:
Each vertex $u$ of $G$ has four copies in $G'$ denoted $u_{A},u_{B},u_{C},u_{D}$. Let $H= 100M$ (sufficiently large {\em even} number). For every $(X,Y)\in \{(A,B),(B,C),(C,D)\}$ and $u,v\in V(G)$, add the edge $(u_{X}, v_{Y})$ of weight $2H+w(u,v)$. For every $u\neq v\in V(G)$, add an edge $(u_{A}, v_{D})$ of weight $5H$. For every $u=v\in V(G)$, add an edge $(u_{A}, u_{D})$ of weight $6H$.

Let $m$ be the number of edges in $G$ and let $W$ be the sum of edge weights in $G$. We claim that WIP($G'$) $- \ W$ is odd iff NTVP($G$) is odd: The sum of $A$-to-$B$ distances is $W+2H\cdot m$. This sum has the same parity as $W$, which we cancel out by subtracting $W$ from WIP($G'$). Notice that the sum of $B$-to-$C$ ($A$-to-$C$) distances and the sum of $C$-to-$D$ ($B$-to-$D$) distances are equal thus by adding both sums the parity of WIP($G'$) does not change. Similarly, the sum of the $X$-to-$X$ distances for every $X\in \{A,B,C,D\}$ is the same and WIP($G'$) is not changed. We are left with the $A$-to-$D$ distances. The sum of $u_{A}$-to-$v_{D}$ distances for $u\neq v$ is $5H\cdot n(n-1)$ (even). If $u$ is not in a negative triangle, then $d(u_{A},u_{D})=6H$ (even) by using the direct edge $(u_{A},u_{D})$. If $u$ is in a negative triangle, the $u_{A}$-to-$u_{D}$ distance is $6H$ plus the weight of the minimum weight triangle of $u$ (odd). Therefore WIP($G'$) $ - \ W$ is odd iff there is an odd number of vertices with a negative triangle. 

\subsection{Negative Triangle Vertex Parity to Sum of Eccentricities Parity}\label{subsec:NTPtoSoE}

The reduction is obtained by tweaking the reduction of Section~\ref{subsec:NTPtoWIPundirected}. We add to $G'$ an additional vertex $y$. For every $u\in V(G)$, we add the edge $(y,u_{D})$ of weight $7H$ and the edges $(y,u_{A}),(y,u_{B})$ and $(y,u_{C})$ each of weight $5H$.

Notice that these changes to $G'$ do not affect the distances between vertices of  $V(G')\backslash \{y\}$ since every path that goes through $y$ has weight of at least $10H$. Recall that $H$ is an even number.
The {\em eccentricity} of a vertex $u$ is defined as $\max_{v}d(u,v)$. The eccentricity of vertices in $B\cup C$ is $5H$ (even), since their distance to $y$ is $5H$ and their distance to any other vertex is bounded by $4H+2M$ (i.e. smaller than $5H$).
The eccentricity of vertices in $D$ is $7H$ (even), since their distance to $y$ is  $7H$ and their distance to any other vertex is bounded by $6H$ (maximized by a vertex in $A$).
The eccentricity of $y$ is $7H$ (even). Finally, the eccentricity of a vertex $u_{A}$ in $A$ is $d(u_{A},u_{D})$, since the $u_{A}$-to-$u_{D}$ distance is at least $6H-3M$ and any other distance is bounded by $5H$ (maximized by $y$ and some $v_{D}$). This means that, as shown in Section~\ref{subsec:NTPtoWIPundirected}, the parity of $\sum_{u}d(u_{A},u_{D})$ equals NTVP($G$).

\subsection{Negative Triangle Vertex Parity to Integer Betweenness Centrality Parity}\label{subsec:NTPtoIntegerBCParity}
The reduction is deterministic and uses a similar graph $G'$ to the one used in the reduction of~\cite{AbboudEtAl} from {\sc Negative Weight Triangle} to {\sc Betweenness Centrality}: Each vertex $u$ of $G$ has four copies in $G'$ denoted $u_{A},u_{B},u_{C},u_{D}$. Let $H= 100M$ (sufficiently large number). For every $(X,Y)\in \{(A,B),(B,C),(C,D)\}$ and $u,v\in V(G)$, add the edge $(u_{X}, v_{Y})$ with weight $2H+w(u,v)$. Add a single vertex $x$ and for every vertex $v\in V(G)$,   add the edges $(u_{A},x),(x,v_{D})$ with weight $3H$. Add two sets of vertices $Z,O$ each of size $ \lceil \log n \rceil $. Let $z_{i}\in Z, o_{i}\in O$ be the $i$'th vertex of the sets. If the $i$'th bit in $u$'s binary representation is $0$, add an edge $(u_{A},z_{i})$ with weight $2H$ and an edge $(o_{i},u_{D})$ with weight $3H$. Otherwise, add an edge $(u_{A},o_{i})$ with weight $2H$ and an edge $(z_{i},u_{D})$ with weight $3H$.      
This dependency on the binary representation assures that every $u_A$ and $v_D$ are connected with a path (of weight $5H$) through $O$ or through $Z$ except for the case where $u=v$. See Figure~\ref{fig:NTtoIBC}.

    \begin{figure}[htb]
        \centering
        \includegraphics[scale=0.3]{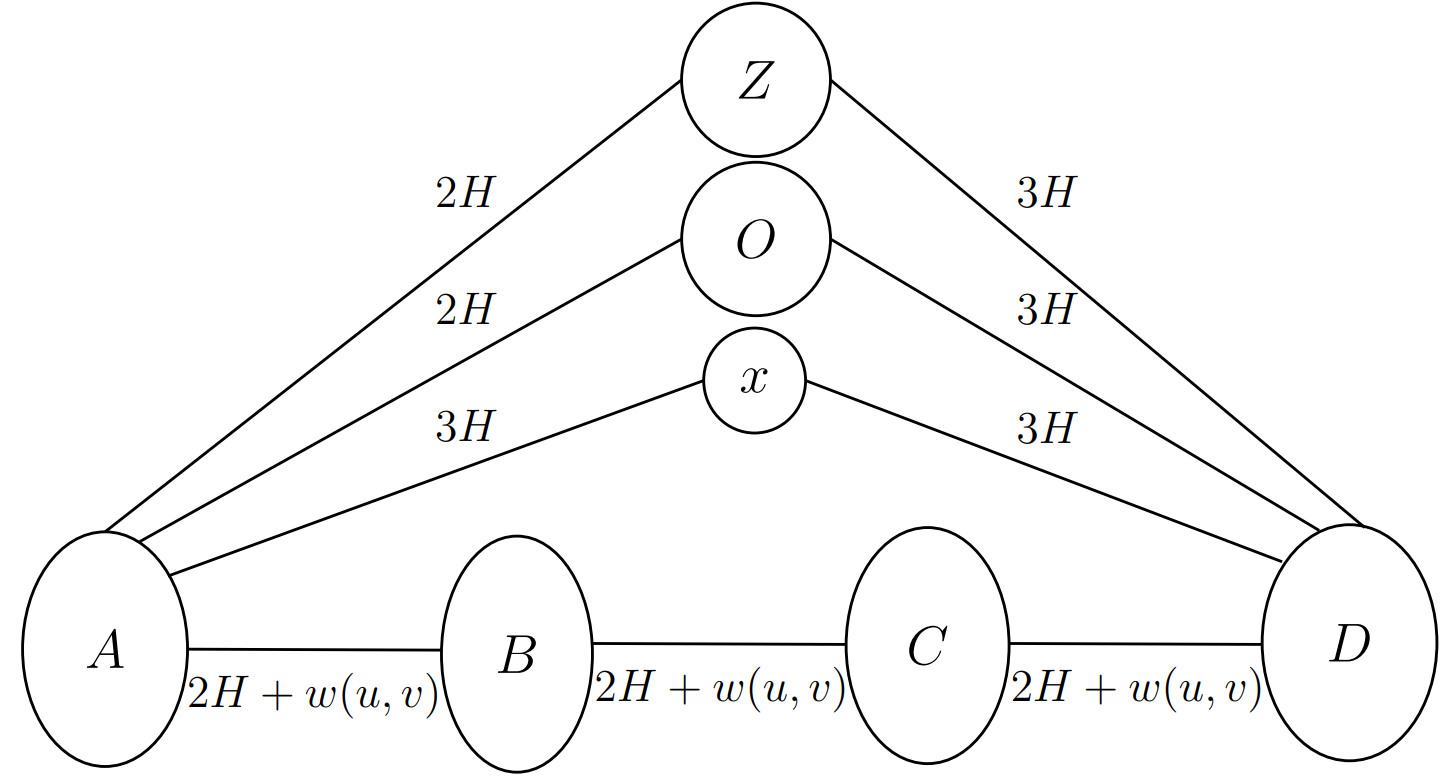}
        \caption{A representation of $G'$ in the reduction from Negative Weight Triangle to Integer Betweenness Centrality Parity}
        \label{fig:NTtoIBC}
    \end{figure}

Consider the Integer Betweenness Centrality Parity of the vertex $x$ in $G'$. 
Assume $n$ is even (otherwise add a vertex to $G$ with no negative triangle). Notice that the only pairs with a shortest path through $x$ can be of the form $(u_{A},u_{D})$ (pairs $(u_{A},v_{D})$ with $v\neq u$ have shorter paths of weight $5H$ through $Z$ or $O$). 
Furthermore, there is a shortest $u_{A}$-to-$u_{D}$ path through $x$ iff $u$ is not in a negative triangle. This is because the distance between $u_{A}$ and $u_{D}$ is the minimum between $6H$ (going through $x$) and $6H+w(u,v)+w(v,t)+w(t,u)$ for some $v,t \in V(G)$. Therefore, the number of pairs $(u_{A},u_{D})$ with shortest paths through $x$ is $n$ minus the number of vertices in a negative triangle, hence the parity of the number of paths going through $x$ in $G'$ is the same as the parity of the number of vertices in $G$ with a negative triangle.

\subsection{APSP to Integer Betweenness Centrality Parity}\label{subsec:APSPtoIntegerBCParity}
We provide a probabilistic (one sided error) reduction from {\sc Negative Weight Triangle} that does not go through {\sc Negative Triangle Vertex Parity}.
We continue from where we stopped in Section~\ref{subsec:NTPtoIntegerBCParity}. Recall that the number of pairs that have a shortest path through $x$ is $n$ minus the number of vertices in a negative triangle. If there is an odd number of pairs then we return that a negative triangle exists. Otherwise, there is an even number of vertices with a negative triangle. We then choose a set $S\subseteq V(G)$ uniformly, and limit $A,D$ to the vertices in $S$ ($B,C$ remain the same). If the number of paths going through $x$ is odd we report that there is a negative triangle, otherwise we report that there is none. If a negative triangle exists, $S$ has an odd number of vertices with a negative triangle with probability $1/2$, and we detect an odd number of pairs. Otherwise, $S$ always has $0$ (even) vertices with a negative triangle, and we succeed with probability $1$. We can repeat the process $O(\log n)$ times and amplify the probability of success to $1- 1/n^{c}$ for any constant $c$.

\section{APSP to Min-Plus Matrix Multiplication Parity}\label{sec:minplusmatrixmul}

\subsection{Min-Plus Multiplication to Min-Plus  Multiplication Parity}\label{subsec:minplusmatrixmultoparity}
Given two $n \times n$ matrices $A$ and $B$ we wish to compute $C=A\otimes B$ where $C[i,j]=\min_k \{A[i,k]+B[k,j]\}$. First assume that for every $i,j$ the value $C[i,j]$ is obtained by a unique index $k$. Let $K$ be the $n \times n$ matrix such that $K[i,j]$ is the unique index $k$ of $C[i,j]$. We show how to compute $K$ by using {\sc Min-Plus Matrix Multiplication Parity}.
 
Define $\hat{A}=2A$, and for any $t\in [\log n]$ define $k_{t}$ as the $t$'th bit of $k$ and $\hat{B_{t}}$ to be the matrix such that $\hat{B_{t}}[k,j]=2B[k,j]+k_{t}$. We compute the parity of $\hat{C_{t}}=\hat{A}\otimes \hat{B_{t}}$ for every $t\in [\log n]$. We claim that the parity of $\hat{C_{t}}[i,j]$ is the $t$'th bit of $K[i,j]$. This is because $\hat{C_{t}}[i,j]= \min_{k}\{2A[i,k]+2B[k,j]+k_{t}\}$. The parity of this value is $0$ if the unique index $k$ that minimizes $A[i,k]+B[k,j]$ has $k_t=0$ and is $1$ otherwise. Therefore, from this parity we can recover the $t$'th bit of $K[i,j]$.

To remove the assumption on the uniqueness of $k$, we define matrices $A'$ and $B'$ as
$A'[i,k]=(n+1) \cdot A[i,k] + k$ and $B'[k,j]=(n+1) \cdot B[k,j]$.
Observe that $A'$ and $B'$ have the uniqueness of $k$ property. This is because if  
$A'[i,k_{1}]+B'[k_{1},j]= A'[i,k_{2}]+B'[k_{2},j]$ for some $k_{1},k_{2}$ then 
$(n+1)\cdot (A[i,k_{1}]+B[k_{1},j]) +k_{1} = (n+1)\cdot (A[i,k_{2}]+B[k_{2},j]) +k_{2}$ and, since $k_{1}$ and $k_{2}$ are smaller than $n+1$, it follows that $k_{1}=k_{2}$. Furthermore, in order to compute $A\otimes B$ it suffices to compute $C'=A'\otimes B'$. Because  if $A'[i,k_{1}]+B'[k_{1},j]\le A'[i,k_{2}]+B'[k_{2},j]$ then (since $k_{1}$ and $k_{2}$ are smaller than $n+1$) $A[i,k_{1}]+B[k_{1},j]\le A[i,k_{2}]+B[k_{2},j]$.
 
As a corollary, we get that APSP is subcubic equivalent to APSP Parity (i.e. the problem of deciding the parity of every pairwise distance in the graph):
Let $M$ be a bound on the absolute values in $A$ and $B$. Create a graph consisting of vertices $a_{i},b_{i},c_{i}$ for every $i\in [n]$ and the edges $(a_{i},b_{j}), (b_{i},c_{j})$ with weights $A[i,j]+3M$ and $B[i,j]+3M$ respectively for every $i,j$. The distance $d(a_{i},c_{j})=6M+\min_{k} \{A[i,k]+B[k,j]\}$ and therefore has the same parity as $(A\otimes B)[i,j]$.

Notice that the reduction can be modified (with a folklore trick) to show that even computing $A\otimes A$ Parity is hard.
Let $D$ be the $n\times n$ matrix with every element equals to $3M$.
Let $E$ be the $2n\times 2n$ matrix $\begin{bmatrix}A & B \\ D & D\end{bmatrix}$.
Then $E\otimes E$ equals $\begin{bmatrix} X & Y\\ Z  & W \end{bmatrix}$ where   $Y=A\otimes B$ since $Y[i,j]=\min\{(A\otimes B)[i,j],(B\otimes D)[i,j]\}=(A\otimes B)[i,j]$.

\subsection{Min-Plus Convolution to Min-Plus Convolution Parity}\label{subsec:MinplusConvToParity}
Given vectors $A$ and $B$ each of length $n$, we wish to compute their convolution $C$ where $C[i]= \min_{i=j+k}\{ A[j]+B[k]\}$. The approach is the same as in Section~\ref{subsec:minplusmatrixmultoparity}. We assume each value $C[i]$ is obtained by a unique index $k$, otherwise we multiply $A$ and $B$ by $n+1$ and add to each $B[k]$ the value $k$ (as in Section~\ref{subsec:minplusmatrixmultoparity}). Let $K$ be the vector such that $K[i]$ is the unique index $k$ of $C[i]$. Define $\hat{A}=2A$, and for any $t\in [\log n]$ define $k_{t}$ is the $t$'th bit of $k$ and $\hat{B_{t}}$ to be the vector such that $\hat{B_{t}}[k]=2B[k]+k_{t}$. Let $\hat{C_{t}}[i]$ be the convolution of $\hat{A}$ and $\hat{B_{t}}$. Then the $t$'th bit of  $K[i]$ is the same as the parity of $\hat{C_{t}}[i]$. This is because $\hat{C_{t}}[i]=\min_{i=j+k}\{2A[j]+2B[k]+k_{t}\}$.

\subsection{Maximum Consecutive Subsums to Maximum Consecutive Subsums Parity}\label{subsec:MaxConsSumToParity}

Given a vector $X$ of length $n$, the maximum consecutive subsums problem asks to compute $\max_{i}\sum_{j=1}^{k} X[i+j]$ for every $k\in [n]$. To achieve this, we first compute (in linear time) the vector $A$ where $A[k]=\sum_{j=1}^{k} X[j]$. The problem then reduces to computing Diff$(A)$ where Diff$(A)[k]=\max_{i}\{A[k+i]-A[i]\}$. In fact, since $X[k]=A[k]-A[k-1]$, there is also a reduction in the opposite direction and so the two problems are equivalent (and their parity versions are equivalent).  
In this section, we show that given the parity of Diff$(A)$ (i.e. the parity of every element in Diff$(A)$) we can compute Diff$(A)$ itself.

Given a vector $A$, we wish to compute Diff$(A)$. We assume that for every $k$, the value Diff$(A)[k]=\max_{i}\{A[k+i]-A[i]\}$ is obtained by a unique index $i$. Otherwise, we multiply every $A[k]$ by $(n^{2}+1)$ and add $k^2$ (similarly to Section~\ref{subsec:minplusmatrixmultoparity}).
Let $I$ be the vector of such unique indices, and let $J$ be the vector where $J[k]= I[k]+k$. By definition, Diff$(A)[k]= A[J[k]]-A[I[k]]$.
We define ${A_{t}}$ to be the vector such that ${A_{t}}[k]= 4\cdot A[k]+k_{t}$ (where $k_{t}$ is the $t$'th bit of $k$). Notice that ${A_{t}}[j]-{A_{t}}[i]=4\cdot(A[j]-A[i])+(j_{t}-i_{t})$ where $(j_{t}-i_{t})\in \{-1,0,1\}$. Thus, for every $k$, Diff$({A})[k]$ is maximized when $j=J[k]$ and $i=I[k]$ (regardless of the values of $j_{t}$ and $i_{t}$). This is because for every $i\neq I[k]$ and $j=k+i$ it holds that $4\cdot(A[j]-A[i])+(j_{t}-i_{t})\leq 4\cdot(A[J[k]]-A[I[k]]-1)+1< 4\cdot(A[J[k]]-A[I[k]])+(J[k]_{t}-I[k]_{t})$.
Observe that that the parity of ${A_{t}}[j]-{A_{t}}[i]$ is $j_{t}\oplus i_{t}$, where $\oplus$ is the bitwise XOR operation.

For every $t\in [\log n]$ we compute the parity of Diff$({A_{t}})$. The computed parity of Diff$({A_{t}})[k]$ is $J[k]_{t}\oplus I[k]_{t}$. Given $J[k]_{1}\oplus I[k]_{1}, \ldots ,  J[k]_{\log n} \oplus I[k]_{\log n}$, we want to compute $J[k]$ and $I[k]$.
Recall that $J[k]=k+I[k]$. Let $c_{1},\ldots, c_{\log n}$ be the carry bits in the binary addition of $I[k]$ and $k$. We know that $I[k]_{t}\oplus k_{t}\oplus c_{t}=J[k]_{t}$ so by substituting $J[k]_{t}\oplus I[k]_{t}$ we compute every $c_{t}=J[k]_{t}\oplus I[k]_{t}\oplus k_{t}$. Given $c_{t},c_{t+1},k_{t},$  and $J[k]_{t}\oplus I[k]_{t}$ we wish to compute $J[k]_{t}$ and $I[k]_{t}$. However, this can only be done when $J[k]_{t}\oplus I[k]_{t}=1$. In this case $I[k]_{t}=\lnot J[k]_{t}= c_{t+1}$. This is because $k_{t}\oplus c_{t}=J[k]_{t}\oplus I[k]_{t}=1$ so $k_{t}+c_{t}=1$ and therefore $c_{t+1}=1$ iff $I[k]_{t}+k_{t}+c_{t}\geq 2$ iff $I[k]_{t}=1$.
We are left with the bits where $J[k]_{t}= I[k]_{t}$.
Let ${A}_{t,p}$ be the vector such that ${A}_{t,p}[k]=4\cdot A[k]+(k_{t}\rightarrow k_{p})$ (compared to ${A_{t}}$, we replace $k_{t}$ with $k_{t}$ implies $k_{p}$). For every $(t,p)\in [\log n]^2$ we compute the parity of Diff$({A}_{t,p})$. The parity of Diff$({A}_{t,p})[k]$ is $(J[k]_{t}\rightarrow J[k]_{p})\oplus (I[k]_{t}\rightarrow I[k]_{p})$ and is denoted as $b_{t,p}$. Given that $J[k]_{t}$ and $I[k]_{t}$ have not been computed yet, we know that $J[k]_{t}= I[k]_{t}$, hence for every $p$ it holds that $b_{t,p}=(I[k]_{t}\rightarrow J[k]_{p})\oplus (I[k]_{t}\rightarrow I[k]_{p})$. Notice that $J[k]>I[k]$ therefore there must be an index $p'$ where $I[k]_{p'}\neq J[k]_{p'}$ (which we previously found) thus $b_{t,p'}=(I[k]_{t}\rightarrow \lnot I[k]_{p'})\oplus (I[k]_{t}\rightarrow  I[k]_{p'})$. Observe that $I[k]_{t}=0$ iff $b_{t,p'}=0$. Overall, we find $I[k]$ for every $k$ using $O(\log^2 n)$ Diff parity computations and $\tilde{O}(n)$ reduction time.

\section{Zero Weight Triangle Counting to Negative Triangle Counting}

In this section we show a simple but surprising reduction from counting zero weight triangles to counting negative triangles.
We show a deterministic reduction from {\sc Zero Weight Triangle} to {\sc Negative Triangle Counting} and a randomized reduction from {\sc Zero Weight Triangle} to {\sc Negative Triangle Parity}.

\subsection{Zero Weight Triangle Counting (Parity) to Negative Triangle Counting (Parity)}\label{subsec:ZWTCountingParityToNegTriangleCountingParity}

We want to count the number of triangles with weight zero in a given graph $G$. Let $\Delta$ be the number of triangles in $G$. We can compute $\Delta$ in matrix-multiplication $O(n^\omega)$ time\footnote{We can also compute $\Delta$ with Negative Triangle Counting by changing every weight in $G$ to $-1$.}. Let $\Delta^{\!0}, \Delta^{\!+}, \Delta^{\!-}$ be the number of zero, positive, and negative weight triangles in $G'$ respectively. Given a subcubic algorithm for {\sc Negative Triangle Counting}, we can compute $\Delta^{\!-}$. By negating weights in $G$ we can compute $\Delta^{\!+}$ as well. Therefore, we can compute $\Delta^{\!0}= \Delta - \Delta^{\!+}- \Delta^{\!-}$. This simple reduction also reduces {\sc Zero Weight Triangle Parity} to {\sc Negative Triangle Parity}.

\subsection{Zero Weight Triangle to Zero Weight Triangle Parity }~\label{subsec:ZWTcountingParity}

Given a graph $G$, we want to find whether there is a zero weight triangle. We create a graph $G'$ as follows: For every vertex $u\in V(G)$, we create three copies $u_{A},u_{B},u_{C}$ in $G'$, and for every edge $(u,v)\in E(G)$ we add the edges $(u_{A},v_{B})$,$(u_{B},v_{C})$, $(u_{C},v_{A})$ to $G'$ (with the same weight as $(u,v)$). Notice that there is a zero weight triangle in $G$ iff there is a zero weight triangle in $G'$.
We now create a graph $G''$ by removing from $G'$ each edge $(u_{B},v_{C})$ with probability $\frac{1}{2}$, and removing from $G'$ each vertex $u_{A}$ with probability $\frac{1}{2}$. We report that a zero weight triangle exists in $G$ iff there is an odd number of zero weight triangles in $G''$. We now show that this  reduction works with probability at least $\frac{1}{4}$.

If there is no zero weight triangle in $G'$, we succeed with probability $1$.
If there is a zero weight triangle in $G'$, then let  $u_{A}$ be a vertex of $G'$ that participates in some zero weight triangle. Since we removed each edge $(v_{B},t_{C})$ with probability $\frac{1}{2}$ then, with probability $\frac{1}{2}$, the vertex $u_{A}$ participates in an odd number of zero weight triangles in $G''$. Let $\Delta^{\!0}_{v_{A}}$ be the number of zero weight triangles in $G''$ that $v_{A}$ participates in. The total number of zero weight triangles in $G''$ is $$\sum_{v\in V(G)}\!\!\! \Delta^{\!0}_{v_{A}} \ \  =  \sum_{v\in V(G)\backslash\{u\}}\!\!\!\!\!\! \Delta^{\!0}_{v_{A}}+\Delta^{\!0}_{u_{A}}.$$ If $\Delta^{\!0}_{u_{A}}$ is odd then, with probability $\frac{1}{2}$, our decision whether to remove $u_{A}$ leads to an odd number of zero weight triangles in $G''$. Overall, the probability of success is therefore at least $\frac{1}{4}$.

\section{Min-Plus Convolution to Knapsack Parity}\label{sec:minplusknapsackparity}

In the {\sc knapsack} problem, given a set of $n$ items $(w_{i},v_{i})$ and a target weight $t$, we wish to pick a multiset of items $I$ that maximizes $\sum_{i\in I} v_{i}$ subject to $\sum_{i\in I} w_{i} \leq t$. When $I$ is required to be a set (and not a multiset) the problem is called {\sc 0/1-knapsack}. In the {\sc Indexed Knapsack} problem, we have $w_{i}=i$ and $t=n$. Finally, the {\sc Coin Change} problem~\cite{DynamicProgramingEquivalences,UnweightedCoinProblem} is the same as {\sc Indexed Knapsack} but with the additional restriction $\sum_{i\in I} i= n$.

The {\sc Knapsack} and the {\sc 0/1-knapsack} problems are equivalent to Min-plus Convolution under randomized reductions~\cite{MinPlusConvEquivalences}.
The {\sc Indexed Knapsack} and the {\sc Coin Change} problem are equivalent to Min-plus Convolution under deterministic reductions~\cite{DynamicProgramingEquivalences}.  
In this section, we show that the parity versions of all the above problems are equivalent to Min-plus Convolution.

\subsection{Super-Additivity Testing to  Knapsack~\cite{MinPlusConvEquivalences}}
 Given a vector $A[0],\ldots,A[n-1]$, the {\sc Super-Additivity testing} problem asks whether $A[i]+A[j]\leq A[i+j]$ for every $i,j$. The problem is subquadratic equivalent to Min-plus Convolution under deterministic reductions~\cite{MinPlusConvEquivalences}. We now give a brief description of the reduction in~\cite{MinPlusConvEquivalences} from {\sc Super-Additivity testing} to {\sc  Knapsack}.
 
First, it is shown in~\cite{MinPlusConvEquivalences} that we can assume without loss of generality that  $0=A[0] < A[1] < \cdots < A[n-1]=M$. Let $D=Mn+1$, the instance of {\sc Knapsack} consists of two types of items: Type-$A$ items are $(i,A[i])$  and Type-$B$ items are $(2n-1-i,D-A[i])$. It remains to show that, when setting $t=2n-1$, the optimal sum of values $\sum_{i\in I} v_{i}$ equals $D$ iff $A$ is super-additive.
 Since $D>\sum_{i}A[i]$, the optimal solution must take at least one Type-$B$ item, and it cannot take more than one because the weight would exceeds $t$. If $A$ is not super-additive, then for some $i,j$ it holds that $A[i]+A[j]> A[i+j]$ and therefore the three items $\{(i,A[i]), (j,A[j]), (2n-1-i-j,D-A[i+j])\}$ constitute a valid solution whose value is larger than $D$. If $A$ is super-additive, then every two Type-$A$ items $(i,A[i]), (j,A[j])$ can be replaced by $(i+j,A[i+j])$ without changing the total weight. Thus, for any $i$, the solution $\{(i,A[i]), (2n-1-i,D-A[i])\}$ is optimal and its value is exactly $D$.
 
 \subsection{Super-Additivity Testing to Knapsack Parity}
 We now modify the above the reduction to obtain a reduction to {\sc Knapsack Parity}. We first remove the item $(0,A[0])$ (since $A[0]=0$ it does not contribute any value). We then replace every Type-$A$ item $(i,A[i])$ by $(i,2A[i])$, every Type-$B$ item $(2n-1-i,D-A[i])$ (with $i\neq 1$) by $(2n-1-i,2(D-A[i]))$, and the Type-$B$ item $(2n-1-1,D-A[1])$ by $(2n-1-1,2(D-A[1])+1)$. We show that $A$ is super-additive iff the value of the optimal solution is odd. 
 
 Once again, every optimal solution must consist of exactly one Type-$B$ item since if there are no Type-$B$ items then the value does not exceed $2D$ as $2D> \sum_{i}2A[i]$, and with more than one Type-$B$ items the weight exceeds $t=2n-1$. If $A$ is not super-additive, then there are $i,j$ such that $k=i+j\geq 2$ and $A[i]+A[j]>A[k]$ therefore the items $\{(i,2A[i]), (j,2A[j]), (2n-1-i-j,2D-2A[k])\}$ constitute a valid solution whose value is larger than $2D+1$. Notice that $k\ge 2$ since if $k=0$ then the total value is $2D$ (thus not optimal), and if $k=1$ then we include the item $(2n-1-1,2(D-A[1])+1)$ and since $t=2n-1$ we can only add the item $(1,2A[1])$ leading to a non-optimal solution with value $2D+1$. Therefore, the optimal solution does not use the item $(2n-1-1,2(D-A[1])+1)$ and hence it has an even value.
 On the other hand, if $A$ is super-additive, then the solution $\{(1,2A[1]),(2n-1-1,2(D-A[1])+1)\}$ is optimal and has value exactly $2D+1$ (odd). This is because among the solutions that include a Type-$B$ item $(2n-1-i-j,2D-2A[k])$ with $k\neq 1$, once again by super-additivity, $\{(k,2A[k]),(2n-1-k,2D-2A[k])\}$ has the maximal value of $2D$ (i.e. smaller than $2D+1$). 
  
 \subsection{Super-Additivity Testing to 0/1-knapsack Parity}
 We now show how to modify the above reductions to be reductions to {\sc 0/1-knapsack} and {\sc 0/1-knapsack Parity}. In the above reductions, when $A$ is super-additive the optimal solution does not use any item more than once, and its total value $V$ is either $D$ or $2D+1$. When $A$ is not super-additive, there is a solution with a higher value than $V$. There is only one case where this solution may use the same item more than once. This happens when $A$ is not super-additive in the following way: $A[i]+A[j]\leq A[i+j]$ for every $i\neq j$ but $A[i]+A[j]> A[i+j]$ for some $i=j$. Therefore, in $O(n)$ time we can check for every $i$ whether $2A[i]\leq A[2i]$ and only if the answer is yes we apply the reduction.

Note that the above reductions also apply to the {\sc Indexed Knapsack Parity} problem. This is because the target weight $t$ equals the total number of items $2n-1$, and each item has a unique weight in $[2n-1]$. The reductions also apply to {\sc Coin Change Parity}: When $A$ is super-additive, the optimal solution for {\sc Coin Change} (which is also an optimal solution for {\sc Knapsack}) has weight $2n-1$ (equal to the number of items) and an odd value ($2D+1$). When $A$ is not super-additive, the optimal solution for {\sc Coin Change} (which is possibly not an optimal solution for {\sc Knapsack}) has weight $2n-1$ (equal to the number of items) and an even value (larger than $2D+1$).

\section{APSP to Sum of Eccentricities} \label{sec:sumofeccentricities}

In this section, we show a subcubic reduction from {\sc Radius} (hence also from APSP) on a graph $G$ to {\sc Sum of Eccentricities} on a graph $G'$. 
Let $R$ be the radius of  $G$. In order to compute $R$ it suffices to find whether $R\geq k$ for any given $k \in [Mn]$ (since then we can binary search for $R$). The constructed graph $G'$ is similar to the one in the reduction of~\cite{AbboudEtAl} from {\sc Diameter} to {\sc Positive Betweeness centrality}: We create $G'$ by multiplying the edge weights of $G$ by 2 and then adding a vertex $x$ and the edges $(x,u)$ and $(u,x)$ each of weight $k$ for every $u\in V(G)$. 
\begin{lemma}
 $R\geq k$ iff the sum of eccentricities of $G'$ is $\sum_{u}\max_{v} {d_{G'}(u,v)}=2kn+k$. 
\end{lemma}
\begin{proof}
This is the same as claiming that $R\geq k$ iff $\sum_{u\neq x}\max_{v} {d_{G'}(u,v)}=2kn$.
If $R\geq k$, then every vertex $u\neq x$ can use $x$ to get to its furthest vertex with a path of length $2k\leq 2R$. Observe that any other path would be of length at least $2R$ (because we have multiplied all edge weights by 2). Therefore, $\sum_{u\neq x}\max_{v} {d_{G'}(u,v)}=2kn$.
If on the other hand $R<k$, then the distance in $G'$ from the radius vertex of $G$ to any other vertex is at most $\max\{2R,k\}<2k$ so this vertex adds less than $2k$ to the sum. All the other vertices add at most $2k$ to the sum, and thus $\sum_{u\neq x}\max_{v} {d_{G'}(u,v)} \neq 2kn$.
\end{proof}

\section{Radius to Radius Parity and  Diameter to Diameter Parity }\label{sec:RadiusAndDiameterParity}
In this section, we show that computing the Radius $R$ (resp. Diameter $D$) of a graph $G$ subcubicaly reduces to computing the parity of $R$ (resp.  $D$). As usual, to compute $R,D$ it  suffices to find whether $R,D\geq k'$ for  $k' \in [Mn]$. Let $k'=(k+1)/2$ for some odd $k\geq 1$. We create a reduction graph $G'$ similarly to~\cite{AbboudEtAl} and to Section~\ref{sec:sumofeccentricities}: We multiply the edge weights of $G$ by $2$ and add a vertex $x$ with $(v,x)$,$(x,v)$ edges of weight $k$ for every $v\in V(G)$. In the case of Radius, we add an additional vertex $y$ with $(v,y)$,$(y,v)$ edges of weight $k$ for every $v\in V(G)$.

Consider first the diameter of $G'$. If $x$ is an endpoint of the diameter of $G'$ then the diameter value is $k$. Otherwise, the diameter value is either $2D$ (by taking the same path as in $G$) or $2k$ (by using a path through $x$). Therefore the diameter of $G'$ has value $\max\{k,\min\{2D,2k\}\}$. 
If $2D\leq k$ (i.e. $D<k'$), then the diameter is $k$ (odd). Otherwise $2D \geq k+1$ (i.e. $D \geq k'$), and the diameter is either $2D$ or $2k$ (even in both cases).

Consider next the radius of $G'$. The radius of $G'$ is the minimum between $2k$ (if $x$ and $y$ are the endpoints) and $\max\{k,\min\{2k,2R\}\}$ (otherwise). If $2R\leq k$, the latter term equals $k$ and the radius of $G'$ is $k$ (odd). Otherwise, $k<2R$ and the radius of $G'$ is either $2k$ or $2R$ (even in both cases). %We now binary search over odd $k$ to find $2R$ and thus $R$.

%\begin{remark}
%We sketch an alternative direct reduction from NWT to Radius Parity (i.e. without going through Radius): Using the reduction of~\cite{AbboudEtAl} from NWT to Radius, we obtain a graph in which the radius is $\min\{6M,6M + \text{weight of minimum weight triangle in } G\}$. Then ensure triangle weights are odd by adjusting weights in the tripartite variant of NWT.
%\end{remark}

\begin{comment}
If the radius vertex of $G'$ is $x$ then its value is $k$. Otherwise, its value is either $\max\{2R,k\}$ (by using the edge to $x$ and the same path as in $G$ to all other vertices) or $2k$ (by using a path through $x$). Therefore, the radius of $G'$ has value $\min\{k,2R\}$. If $2R\geq k+1$ (i.e $R\ge k'$), the radius is $k$ (odd). Otherwise $2R<k+1$ (i.e $R<k'$) and the radius is $2R$ (even).  
\end{comment}

\section{Diameter to Reach Centrality Parity}\label{sec: ReachCentralityParity}
Assume the diameter $D$ is even by multiplying the edge weights by $2$. 
We want to be able to answer whether $D\geq k$ for an {\em even} $k$, because then we can find $D$ by binary search over $[Mn]$ (although $k$ is even, we can find $D$ since it is even as well).
The original reduction of~\cite{AbboudEtAl} from {\sc Diameter} to {\sc Reach Centrality} (RC) uses the same graph $G'$ from Section~\ref{sec:sumofeccentricities}. I.e. $G'$ is obtained by (again) multiplying the edge weights of $G$ by 2 and then adding a vertex $x$ and the edges $(x,u)$ and $(u,x)$ each of weight $k$ for every $u\in V(G)$.
If $D\geq k$, then RC($x$)$=k$ (even) since the path through $x$ (of weight $2k$) is not longer than the shortest path within $G$ (of weight $2D$). If $D<k$, there is no shortest path going through $x$ and thus RC($x$) is $0$ (also even). To avoid this, we add a vertex $y$ with the bidirectional edge $(x,y)$ of weight $1$. Now, if  $D\geq k$ then we still have RC($x$)$=k$ (even) but when $D<k$ we now have RC($x$)$=1$ (odd) because the only shortest paths through $x$ consist of two edges $(u,x)$ of weight $k$ and $(x,y)$ of weight 1.

\section{Minimum Weight Triangle Parity}
The {\sc Minimum Weight Triangle Parity} problem asks for the parity of the weight of the minimum weight triangle. In Section~\ref{subsec:APSPtoMinimumWeightTriangleParity} we show that this problem is subcubic equivalent to APSP. We then use this equivalence in Sections~\ref{subsec:MaxSubParity} and~\ref{subsec:ReplacementPath} to reduce APSP to the parity versions of {\sc Maximum Subarray}, {\sc Replacement Paths}, and {\sc Second shortest path}.

\subsection{APSP to Minimum Weight Triangle Parity}\label{subsec:APSPtoMinimumWeightTriangleParity}
The reduction is from {\sc Negative Weight Triangle}. Given an instance $G$ of {\sc Negative Weight Triangle}, we first multiply the edge  weights by $4$ and then add $1$ to each edge. Then, we add a triangle of weight $0$. The resulting graph is $G'$. If a negative triangle exists in $G$, then the minimum weight triangle of $G'$ has an odd weight. If a negative triangle does not exist in $G$, then the minimum weight triangle in $G'$ has an even weight (zero).

We can modify the above reduction to obtain a property that will be useful later on. Instead of adding an arbitrary zero weight triangle, for every vertex $u$ in $G$ we add two copies $u_{1},u_{2}$ and add the bidirectional edges $(u,u_{1}),(u,u_{2}),(u_{1},u_{2})$ with weight $0$. Notice that for every vertex $v\in V(G')$, the parity of the minimum weight triangle that $v$ participates in is even (weight $0$) if $v$ is not in a negative triangle or $v \notin V(G)$, and is odd  if $v\in V(G)$ and $v$ is in a negative
triangle. We now have the property that there is no negative triangle in $G$ iff the minimum weight triangle of {\em every} vertex of $G'$ is even (and not only the minimum weight triangle of $G'$). 
 
\subsection{Minimum Weight Triangle Parity to Maximum Subarray Parity}\label{subsec:MaxSubParity}
We use the existing reduction of~\cite{MaxSubArray} from {\sc Minimum Weight Triangle} to {\sc Maximum Subarray}. 
Their reduction creates an instance of {\sc Maximum Subarray} with a weight of $110M$ minus the weight of the minimum weight triangle. We now ask for the parity of the maximum subarray, which is the same as the parity of the minimum weight triangle.

\subsection{Minimum Weight Triangle Parity to Replacement Paths Parity and to Second Shortest Path Parity}\label{subsec:ReplacementPath}
We use the reduction of~\cite{WilliamsNegativeTriangle} from {\sc Negative Weight Triangle} on a graph $G$ to  {\sc Replacement Paths} on a graph $G'$. We assume $M$ is even ($M$ only serves as an upper bound on the edge weights). Given a shortest path $P$ and an edge $e=(u,v)$ on $P$, a {\em detour} of $e$ is  a $u$-to-$v$ path that is internally disjoint from $P$. Let $v_{1}, \dots, v_{n}$ be  the vertices of $G$. The reduction of~\cite{WilliamsNegativeTriangle} constructs a graph $G'$ that includes a shortest path $p_{0}$-$p_1$-$\cdots$-$p_{n}$ such that the replacement path of $e_{i}=(p_{i-1},p_{i})$ (i.e. the shortest $p_{0}$-to-$p_{n}$ path in $G'$ that avoids $e_{i}$) has the following properties: (1) it is composed of the prefix $p_{0}$-$p_1$-$\cdots$-$p_{i-1}$, the minimum weight detour of $e_{i}$ in $G'$, and the suffix $p_{i}$-$p_{i+1}$-$\cdots$-$p_{n}$, (2) its weight is $Mn$ plus the weight of the minimum weight triangle of $v_{i}$ in $G$. Notice that the weight of the prefix $d(p_{0},p_{i-1})$ and suffix $d(p_{i},p_{n})$ is known so the replacement paths parity provides the parity of the minimum weight detour of every $e_i$. Since $Mn$ is even, the parity of the minimum weight detour of $e_{i}$ is the same as the parity of the minimum weight triangle of $v_{i}$. This completes the reduction because, by the property achieved in Section~\ref{subsec:APSPtoMinimumWeightTriangleParity}, it suffices to find if there is a vertex $v_{i}$ with an odd minimum weight triangle. 

A similar reduction works for the {\sc Second Shortest Path} problem, since (as shown in~\cite{WilliamsNegativeTriangle}) the weight of the second shortest $p_{0}$-to-$p_{n}$ path in $G'$ is $Mn$ plus the weight of the minimum weight triangle in $G$. Thus, the parity of the second shortest path is the same as the parity of the minimum weight triangle.

\section{Co-Negative Triangle to Max Row Sum Parity}\label{sec:maxrowsum}
\vspace{-3mm}

In this section, given an instance $G$ to the {\sc Co-Negative Triangle} problem, we show a subcubic reduction that generates an instance $G'$ to the {\sc Max Row Sum Parity} problem. The graph $G'$ is similar to the one in Section~\ref{sec:MedianParity} except that: (1) it is now a directed graph (directed as in Figure~\ref{fig:MaxRowSumFig}), and (2) it includes the additional edges $(u_{B'},v_{C})$ of weight $4H$, $(u_{C'},v_{C})$ of weight $4H$, all other edges (including edges between vertices in the same set) have weight $H$.

\begin{figure}[htb]
    \centering
    \includegraphics[scale=0.38]{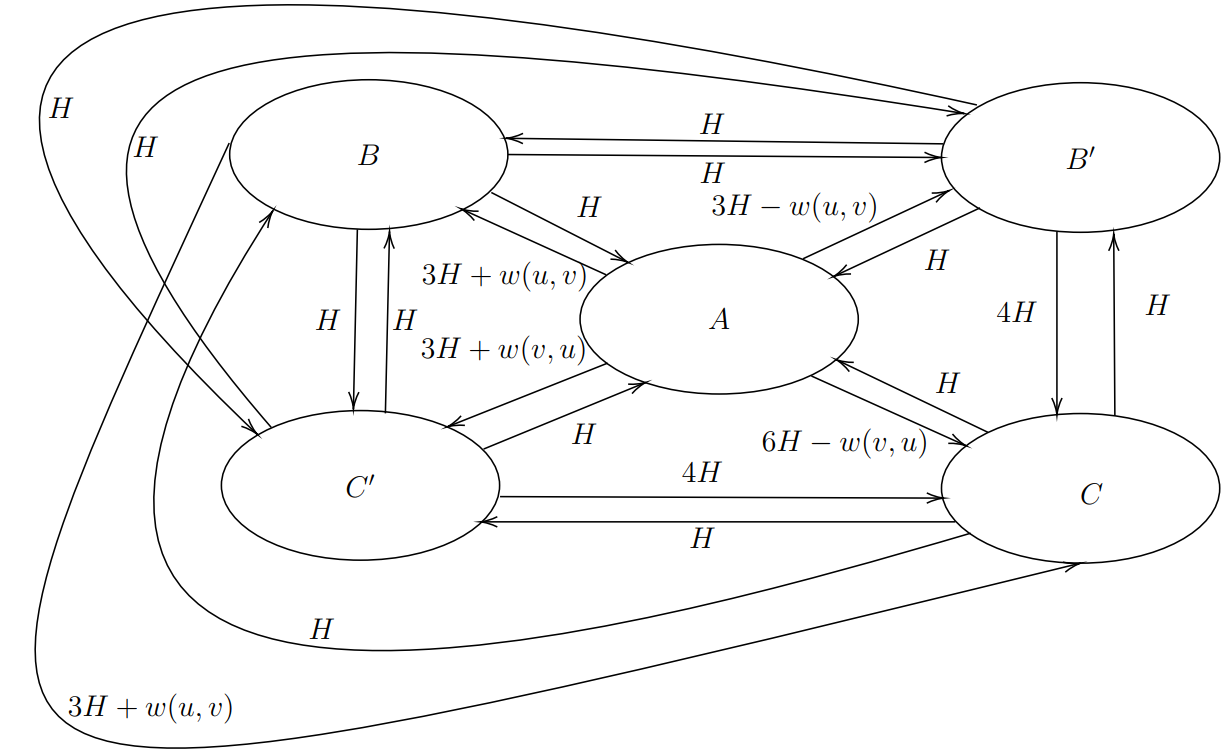}
    \caption{The graph $G'$ in the {\sc Co-Negative triangle} to {\sc Max Row Sum} reduction. Every two vertices in the same set are connected with bidirectional edges of weight $H$.}
    \label{fig:MaxRowSumFig}
\end{figure}

\begin{lemma}
There exists a vertex in $G$ that does not belong to any  negative triangle iff the max row sum of $G'$ equals exactly $(16n-1)H$.
\end{lemma}
 \begin{proof}
For a vertex $u_{X}$ with $X\neq A$,  its sum of distances $\sum_{v\in V(G')}{d_{G'}(u_{X},v)}$ is smaller than $(8n-1)H$ (it is maximized when $X=B'$ or $X=C'$). For $u_{A}$, as shown in Section~\ref{sec:Negative Weight Triangle to Median}, the sum is exactly $(16n-1)H+ \sum_{v\in V(G)} F(u,v)$ where $F(u,v)= 
\min\{0, \min_{t\in V(G)}\{w(v,u)+w(u,t)+w(t,v)\}\}$. Since $H$ is large enough, the sum for $u_{A}$ is in $[(15n-1)H$,$(16n-1)H]$ (because $\sum_{v\in V(G)} F(u,v) \geq -3Mn >-Hn$) so the vertex with the maximal sum comes from $A$ and its sum is $(16n-1)H$ iff there is a vertex $u$ such that $\sum_{v} F(u,v)=0$ (i.e. $u$ does not belong to any negative triangle).
\end{proof}

Using the above claim, we now show how to find a vertex that does not belong to any negative triangle in $G$ (or report that no such vertex exists) using $O(\log n)$ executions of {\sc Max Row Sum Parity}. 
Assume $n$ is odd (otherwise add 3 vertices forming a negative triangle to $G$). Multiply all the edge weights of $G'$ by $4n$ (notice that this makes the max row sum value $(16n-1)H\cdot 4n$). Then, given a set of vertices $T\subseteq A$, (initiated to be $A$), pick an arbitrary subset $S$ of $T$ of half of its size. Temporarily subtract $1$ from the  $S$-to-$C$ edges. Now solve {\sc Max Row Sum Parity} on $G'$. If the maximal row sum is even, set $T\gets S$.  If the maximal row sum is odd, set $T\gets T/S$. We continue recursively for $O(\log n)$ steps until $T$ contains a single vertex. We then check if this vertex participates in a negative triangle in $G$. If it doesn't, then we conclude that no such vertex exists.

In contrary to the procedure of Section~\ref{subsec:NegativeTriangleToMedianParity}, the above procedure finds the vertex with the max row sum only when its value is $(16n-1)H$ (i.e. when there exists a vertex $u$ with no negative triangle). If $u_{A}\in T\backslash S$, then the max row sum value remains $4n\cdot(16n-1)H$ (even). If $u_{A}\in S$, then every distance from $u_{A}$ to a vertex $v_{C}$  decreases by $1$ (for every vertex $v_{C}$ the $u_{A}$-to-$v_{C}$ shortest path is the direct edge $(u_{A},v_{C})$) and its max row sum value is $4n\cdot(16n-1)H-n$ (odd).

\section{Min-plus Convolution to Tree Sparsity Parity}\label{sec:TreeSparsity}
The {\sc Tree Sparsity} problem is, given a node-weighted tree, to find the maximum weight of a subtree of size $k$. 
The {\sc Sum$_3$} problem~\cite{TreeSparsityToMinPlusConv} is, given three $n$-length vectors $A,B,C$ whose elements are integers in $[-M,M]$, to decide if there are $i,j$ such that $A[i]+B[j]+C[i+j]\geq 0$. 
In~\cite{TreeSparsityToMinPlusConv}, a deterministic subquadratic reduction is shown from {\sc Min-Plus convolution} to {\sc Tree Sparsity}: First, they give a (subquadratic) reduction from {\sc Min-Plus convolution} to {\sc Sum$_3$}. Then, they give a (subquadratic) reduction from {\sc Sum$_3$} to {\sc Tree Sparsity}.
Their reduction from {\sc Sum$_3$} to {\sc Tree Sparsity} creates an instance of {\sc Tree Sparsity} of size $O(n)$ in which the maximum weight subtree of size $k=n+2$ has weight $300M+10M(n+2)+\max_{i,j}\{A[i]+B[j]+C[i+j]\}$. 
We next show how to modify the {\sc Sum$_3$} to {\sc Tree Sparsity} reduction so that it reduces to {\sc Tree Sparsity Parity}.

For a subset $S \subseteq [n]$ let $f(S,i)=1$ if $i\in S$ and $0$ otherwise. We define the vectors $A_{S},B_{S}, C_{S}$ as follows: $B_{S}[i]=2B[i], C_{S}[i]=2C[i]$, and $A_{S}[i]=2A[i]+f(S,i)$. Similarly to the reduction in Section~\ref{subsec:NegativeTriangleToMedianParity}, we use a recursive algorithm with $O(\log n)$ iterations in order to find the index $i$ that maximizes the sum condition ($\max_{i,j}\{A[i]+B[j]+C[i+j]\}$). We initialize $S \leftarrow [n]$. This $S$ clearly   contains the index $i$ that maximizes the sum condition. We then arbitrarily choose a set $S'\subseteq S$ of half the size of $S$ and apply the reduction of~\cite{TreeSparsityToMinPlusConv} from {\sc Sum$_3$} to {\sc Tree Sparsity}, but with {\sc Tree Sparsity Parity} instead. Namely, we apply the reduction with the vectors $A_{S'}, B_{S'}, C_{S'}$ and obtain the parity of $\max_{i,j}\{A_{S'}[i]+B_{S'}[j]+C_{S'}[i+j]\}$. If it is odd, then there is some index $i\in S'$ that maximizes $\max_{i,j}\{2(A[i]+B[j]+C[i+j])+f(S',i)\}$ so we set $S\leftarrow S'$. 
If it is even, then there is no $i\in S'$ that maximizes the sum condition so we set $S\leftarrow S/S'$. We continue recursively for $O(\log n)$ steps until $S$ contains a single index $i'$. We then compute $\max_{j}\{A[i']+B[j]+C[i'+j]\}$ in $O(n)$ time and check whether the value is negative.

\section*{Acknowledgments}
We thank Rose Bader and Noam Licht for pointing out a mistake in an earlier version of Section~\ref{sec:RadiusAndDiameterParity}.

\end{document}